\documentclass[runningheads]{llncs}
\usepackage{graphicx}
\usepackage{amsmath}
\usepackage{tabularray}
\usepackage{booktabs}
\usepackage{rotating}
\usepackage[misc,geometry]{ifsym}
\usepackage{enumitem}
\begin{document}
\title{Disease Incidence in a Stochastic SVIRS Model with Waning Immunity\thanks{Supported by Ministry of Science and Innovation (Government of Spain), Project PID2021-125871NB-I00. D. Taipe acknowledges the support of Banco Santander and Complutense University of Madrid, Pre-doctoral Researcher Contract CT63/19-CT64/19.}}
\titlerunning{Disease Incidence in a Stochastic SVIRS Model with Loss of Protection}
%
\author{M.J. Lopez-Herrero \Letter\inst{1}\orcidID{0000-0003-2835-2940} \and
D. Taipe\inst{1}\orcidID{0000-0002-2207-4129}  }
\authorrunning{Lopez-Herrero and Taipe}
%
\institute{Complutense University of Madrid, Madrid, Spain \\
\email{\{lherrero \Letter, dtaipe\}@ucm.es}\\
}
\maketitle              
\begin{abstract}
This paper deals with the long-term behaviour and incidence of a vaccine-preventable contact disease, under the assumption that both vaccine protection and immunity after recovery are not lifelong. The mathematical model is developed in a stochastic markovian framework. The evolution of the disease in a finite population is thus represented by a three-dimensional continuous-time Markov chain, which is versatile enough to be able to compensate for the loss of protection by including vaccination before the onset of the outbreak and also during the course of the epidemics.

\end{abstract}
\section{Introduction}

An essential tool to represent the evolution of an infectious process through a population is mathematical modeling. Many epidemic models involve a compartmental division of individuals which take into account their health status with respect to the infection. It is also a common assumption that individuals are randomly in contact with each other and have no preferences for relationship. 

The fundamental compartmental model for studying an infectious disease that confers immunity is the SIR proposed by Kermack \& McKendrick \cite{KMc}, which can be modified to take into account the control of disease spread (e.g., vaccination \cite{Angelovetal,aaa,Zh}) or the effects of loss of immunity \cite{Angelovetal,aaa,bbb}.

In this paper, we consider an infectious disease taking place in a finite and isolated population. We assume that the disease is transmitted by direct contact with an infectious individual. After recovery, individuals show temporary immunity. In addition, we assume that the disease is a vaccine-preventable infection, but vaccinated individuals are not lifelong fully protected, either because of vaccine failures or because of waning immunity. Given the latter fact, re-vaccination of susceptible individuals is planned as the epidemic progresses.

Consequently, the involved compartmental model divides the population in four classes: Susceptible to the infection (S), vaccine protected (V), infectious (I) and recovered-temporary immune individuals (R). Once the temporary immunity is lost, individuals become susceptible to the disease again. Figure \ref{fgtr} shows the movement of individuals between the four compartments involved in the mathematical model.


Our research will focus on the study of the evolution of a communicable disease by incorporating vaccination and loss of immunity into the epidemiological model, considering both waning vaccine effects and temporary immunity after recovery.
In particular, we consider a Markovian model to represent the spread of the pathogen that causes the disease in a finite population of constant size $N$. 

\begin{figure}
    \centering
    \includegraphics[width=0.86\textwidth]{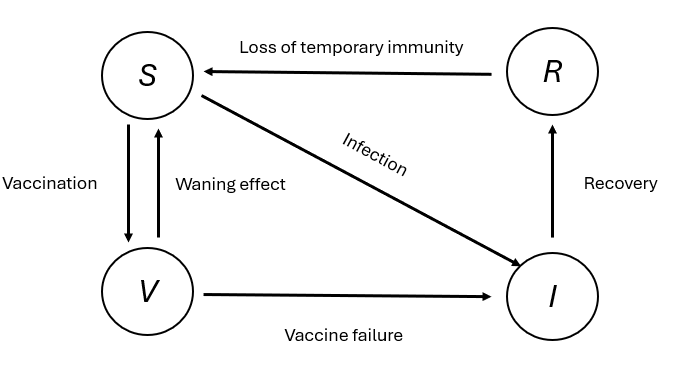}
    \caption{Epidemic compartmental diagram}
    \label{fgtr}
\end{figure}

The aim of this research is, firstly, to observe the long-term behaviour of the Markov chain and thus of the epidemic process itself. Secondly, to study the total number of infections that occur during an outbreak (i.e., the period from the appearance of the first case to the time when there are not infectious individuals in the population).

Theoretical derivations related to the stationary distribution and the probabilistic description of the disease incidence during an outbreak will result from the application of the first-step methodology \cite{ALH14,A-LH}. 

\section{Model description}

Mathematical description involves a continuous time Markov chain (CTMC), which provides the evolution of the disease in terms of the number of individuals present in each compartment at any time $t$. Hence, we record the number of unprotected susceptible, $S(t)$, vaccinated, $V(t)$, infected, $I(t)$, and recovered-temporary immune individuals, $R(t)$. The hypothesis of constant size for the population gives, for any $t \geq 0$, the following relationship among the sizes of the allowed compartments: $N=S(t)+V(t)+I(t)+R(t).$
In consequence, the evolution of the epidemics is represented by a three-dimensional CTMC
\[
\mathcal{X}=\{ (I(t),S(t), V(t)): t\geq 0\}
\]
with state space $\mathcal{S}= \{(i,s,v): 0 \leq i,s,v \leq N, 0\leq i+s+v \leq N\}$, containing a possibly large number of $(N+1)(N+2)(N+3)/6$ states.

We assume that individuals have no relationship preferences, so that any individual can be in contact with any other inhabitant of the population. 

The pathogenic agent is transmitted to susceptible population by direct contact with an infectious individual, at time point of a time-homogeneous Poisson process with the usual bilinear mass action form $\beta(i,s)=\beta i s$, depending on infectious and susceptible individuals, where $\beta$ represents the effective contact rate. Since the vaccine administered is not perfect, we assume that there is a probability, $h$, that it fails in protecting vaccinated individuals. Hence, the pathogen can be transmitted to vaccinated individuals when they are in touch with the infectious ones. The transmission function in the vaccinated class is $\eta(i,v)=\beta h iv$, that depends on the existing infectious and vaccinated individuals, and also on the probability of vaccine failure.

Infectious periods of different infected individuals are represented by independent exponentially distributed random variables, with rate $\gamma$. Each recovered individual, no matter if he was previously vaccinated or not, develops a temporal immunity. The length of the immunity periods of the recovered individuals are independent and also exponentially distributed with rate $\epsilon$. When temporal immunity disappears, recovered individuals become susceptible to the disease. 
To prevent the spread of the disease, susceptible individuals can receive the available vaccine. The vaccination process is scheduled at time points of a time-homogeneous Poisson process with rate $\rho$. Finally, for each vaccinated individual, vaccine protection is assumed to last for an exponentially distributed random time with rate $\theta$.

According the above description, transitions between states in $\mathcal{S}$ are consequence of one of the following events:
\begin{itemize}
    \item[$a_1:$] New infection of a susceptible individual
    \item[$a_2:$] New infection of a vaccinated individual, due to vaccine failure
    \item[$a_3:$] Recovery of one of the existing infectious individuals
     \item[$a_4:$] Vaccination of a susceptible individual
     \item[$a_5:$] A vaccinated person loses vaccine protection
     \item[$a_6:$] A recovered individual loses temporary immunity
\end{itemize}
Given an initial state $(i,s,v) \in \mathcal{S}$, Table \ref{tab1} displays information about possible transitions from the initial state, providing the final states and rates associated with each of the events $a_1 - a_6$.

\begin{table}
\caption{Transitions and rates for each effective event.}\label{tab1}
\centering
\begin{tabular}{|c|c|c|}
\hline
Effective event &  Final state & Rate\\
\hline
$a_1$ &  $(i+1,s-1,v)$ & $\beta i s$\\
$a_2$ &  $(i+1,s,v-1)$ & $h \beta i v$\\
$a_3$ &  $(i-1,s,v)$ & $\gamma i $\\
$a_4$ &  $(i,s-1,v+1)$ & $ \rho s$\\
$a_5$ &  $(i,s+1,v-1)$ & $ \theta v$\\
$a_6$ &  $(i,s+1,v)$ & $ \epsilon (N-i-s-v)$\\
\hline
\end{tabular}
\end{table}

 Sojourn times at each state in $\mathcal{S}$ are independent and exponentially distributed random variables, with rate
\begin{equation}
    q_{(i,v,s)}= \beta i (s+h v)+\gamma i +\rho s+ \theta v+\epsilon (N-i-s-v). \label{sr}
\end{equation}

Since we assume that the population is isolated, part of the states in $\mathcal{S}$ - those corresponding to situations where there are neither infectious nor recovered individuals - are absorbing. Let us denote by $\mathcal{S}_A=\{ (0,s,N-s): 0\leq s\leq N  \}$ the irreducible set of absorbing states. On the other hand, the set of transient states is $\mathcal{S}_T= \mathcal{S}-\mathcal{S}_A$, which is reducible and finite.

\subsection{Long term behavior} \label{s2.1}
As $\mathcal{X}$ is a finite state CTMC, the absorption into $\mathcal{S}_A$ is certain and it occurs in a finite expected time. Moreover, since the absorbing set is a single class of communicating states, the stationary distribution assigns mass to every state in $\mathcal{S}_A$. We notice that our model is a special case of a finite birth-death process (see, for instance \cite{Ross}, Chapter 6), whose stationary distribution is well known. Focusing on the number of susceptible individuals at time t, a birth corresponds to the end of vaccine protection, due to waning immunity; a death corresponds to the vaccination of a susceptible individual. Thus, given that the process $\{S(t):t>0\}$ is in state $s$, the birth and death rates correspond to $\theta (N-s)$ and $\rho s$, respectively.  After some algebra we derive its stationary distribution,  
 $\{p_s: 0\leq s \leq N\}$, which is given by a binomial law
\begin{equation*}
    p_s= \binom{N}{s}\left(\frac{1}{1+\rho / \theta}\right)^{s}   \left(\frac{\rho /\theta}{1+\rho / \theta}\right)^{N-s}, \text{ for }0\leq s \leq N.
\end{equation*}
Hence, we got for any state $(i,s,v)\in \mathcal{S}$ that 
\begin{equation*}
    lim_{t \xrightarrow{} \infty}P\{ I(t)=i, S(t)=s, V(t)=v \}= \delta_{i,0} \delta_{v, N-s}  \binom{N}{s} \left(\frac{\theta}{\theta + \rho}\right)^s \left(\frac{\rho}{\theta + \rho}\right)^v ,
\end{equation*}
where symbol $\delta_{a,b}$ is the Kronecker's delta function, defined as $1 $ when $a=b$ and by $0$ otherwise.

\section{Incidence of the disease during an outbreak}

This section examines the number of infections that occur during an outbreak of the disease. According to the model description, the epidemic stops as soon as there are not infectious cases in the population. We will then consider that outbreaks last while infectious individuals are present in the population and we assume that the outbreaks start from a single infectious individual.

To study the incidence of cases of infection during an outbreak, we will partition the whole state space in levels according to the number of infectious individuals in the population:
\begin{equation*}
    \mathcal{S}= \bigcup_{i=0} ^N \mathcal{S}( i),
\end{equation*}
where, for $0 \leq i \leq N $, each set $\mathcal{S}(i)=\{(i,s,v) \in \mathcal{S}: 0\leq s+v \leq N-i   \}$ contains $c_i= \binom{N-i+2}{2}$ states.

For later use, the above partition of $\mathcal{S}$ by levels, as well as the transition rates among them yield a block-tridiagonal structure for the infinitesimal generator of $\mathcal{X}$ that look as follows:
    \begin{eqnarray}
    {\bf Q} &=& \left(\begin{array}{ccccc}
        {\bf A}_{0,0} & {\bf A}_{0,1} & & &
        \\
        {\bf A}_{1,0} & {\bf A}_{1,1} & {\bf A}_{1,2} & &
        \\
                      & \ddots        & \ddots        & \ddots &
        \\
                      &               & {\bf A}_{N-1,N-2} & {\bf A}_{N-1,N-1} & {\bf A}_{N-1,N}
        \\
                      &               &                   & {\bf A}_{N,N-1}   & {\bf A}_{N,N}
        \end{array}\right)
        \label{eq:genQ}
    \end{eqnarray}
where sub-matrices ${\bf A}_{i,i^*}$ refer to the transition rates from states in the level $S(i)$ to states in the level $S(i^*)$. The diagonal entries correspond to $-q(i, s,v)$, defined in equation (\ref{sr}), and the remaining rates are summarised in Table \ref{tab1}.

In terms of the CTMC describing the evolution of the epidemics, the outbreak starts from a state $(1, s_0,v_0)$ 
and ends when the chain enters into the set $\mathcal{S}(0)$.
\vspace{0.1in}

Let us denote by $L$, the random variable that records the cases of infection observed during the outbreak. We will describe its probabilistic behavior with the help of a set of auxiliary variables, namely $\{L_{(i,s,v)}: (i,s,v) \in \mathcal{S}  \}$. Where each auxiliary variable $L_{(i,s,v)}$ is defined as the number of new infections that occur during the remaining part of the outbreak, given that the current state of the Markov chain is $(i,s,v)$.

To analyze the total incidence during the outbreak we take into account the following relationship:
\begin{equation}
    L=1+L_{(1,s_0,v_0)}.\label{L1}
\end{equation}

Next, we introduce some notation regarding mass and generating functions, and factorial moments of the auxiliary variables, conditioned to a specific state $(i,s,v) \in \mathcal{S}$.
\begin{align*}
     y^j_{(i,s,v)} &= P\{ L_{(i,s,v)}=j  \}, \text{ for } j\geq 0, \\
       \varphi_{(i,s,v)}(z) &= E [z^{L_{(i,s,v)}}]=\sum_{j=0}^{\infty} z^j y_{(i,s,v)}^j 
       , \text{ for } \vert z \vert \leq 1, 
  \end{align*}   
  \[
  m_{(i,s,v)}^k
=       \left\{
\begin{array}{cc}
  \sum_{j=0}^{\infty} y_{(i,s,v)}^j ,  & \text{ for } k=0, \\
    E[\prod_{j=0}^{k-1}(L_{(i,s,v)}-j)],  & \text{ for } k\geq 1.
\end{array}
  \right.     
  \]

We notice that for states $(0,s,v)\in \mathcal{S}(0)$ no additional infections are possible. So, we get
\begin{equation}
    y_{(0,s,v)}^j=\delta_{0,j}, \text{ for }j\geq 0.
\end{equation}
Consequently,
\begin{align}
    \varphi_{(0,s,v)}(z)&=1, \text{ for } \vert z \vert \leq 1, \label{phi0}\\
         m_{(0,s,v)}^k&=\delta_{0,k}, \text{ for } k \geq 0 \label{m1},
\end{align}
where the result in (\ref{m1}) comes from equation (\ref{phi0}) and the relationship $m_{(i,s,v)}^k= \frac{\partial ^k \varphi_{(i,s,v)}(z)}{\partial z^k}\vert_{z=1} $, for $(i,s,v) \in \mathcal{S}$ and $k \geq 0$.

For the remaining group of states $(i,s,v) \in \mathcal{S}-\mathcal{S}(0) $ we will find recursive results to determine probabilities, generating functions and moments.
\vspace{0.15in}

To study the mass distribution of $L$, we assume that the outbreak starts from the state $(1,s_0,v_0)$, where $1+s_0+v_0=N$. 

Given a non-negative integer $k$ and a state $ (i,s,v) \in \mathcal{S}$, we introduce the tail probability $x_{(i,s,v)}^k= P\{ L_{(i,s,v)} \geq k \}$, which represent the probability of having at least $k$ new contagions in the rest of the outbreak, given that the current state of the population is $(i,s,v)$.

Then, we have from equation (\ref{L1}) that
\begin{equation}
    P\{ L\geq k \}=P\{L_{(1, s_0,v_0)}\geq k-1  \}= x_{(1,s_0,v_0)}^{k-1}, \text{ for } k\geq 1.
\end{equation}
As all the epidemic outbreaks start from a single infectious individual, the event $\{ L\geq 1\}$ (or equivalently, $\{ L_{(1,s_0,v_0)} \geq {0} \}$) occurs almost surely. So, the tail probabilities of $0$ additional infections in the outbreak are $x_{(1,s_0,v_0)}^{0}=P\{ L\geq 1 \}=P\{L_{(1, s_0,v_0)}\geq 0  \}= 1$.

We can extend trivially the above result for tail probabilities of $0$ additional infections conditioned to any state $(i,s,v) \in \mathcal{S}$. That is,
\begin{equation}
  x_{(i,s,v)}^0=1.  \label{tp0} 
\end{equation}

On the other hand, as the outbreak ends when the chain $\mathcal{X} $ reaches any state in $\mathcal{S}(0)$, for states in this level we have that
\begin{equation}
    x_{(0,s,v)}^k=0, \text{ for } k \geq 1.\label{tp1}
    \end{equation}
Given $k\geq 0$, let us denote by $\mathbf{x}_i^k$ the $c_i$-dimensional vector whose components are the tail probability of $k$ additional infections related to states in the level $\mathcal{S}(i)$. That is, $\mathbf{x}_i^k=(x^k_{(i,s,v)}: (i,s,v) \in \mathcal{S}(i))^{\prime}$, where the symbol $\prime$ means transpose.

Moreover, in what follows $\mathbf{1}_a$ and $\mathbf{0}_a$ represent the all ones and all zeros vectors of dimension $a$, respectively, and the empty products, appearing in theoretical derivations, are an identity matrix of the appropriate dimension.

The following theorem provides a scheme to determine the tail probabilities recursively.

 \begin{theorem} \label{Teor1}
       For a given integer $ k \geq 1$, the tail probability vectors  $\mathbf{x}_i^k$ are recursively determined from  $\mathbf{x}_i^{k-1}$ through the following equations
\begin{eqnarray} \mathbf{x}_i^0 &=&\mathbf{1}_{c_i} \text{, for }\;0 \leq i \leq N, \label{T1a}\\
\mathbf{x}_0^k &=&\mathbf{0}_{c_0} \text{, for }\; k \geq 1, \label{T1b}\\
\mathbf{x}_i^k &=& - \mathbf{A}_{i,i}^{-1}\sum _{m =1}^{i} \left({\displaystyle\prod _{j=m}^{i -1}}  \mathbf{A}_{j+1,j} (-\mathbf{A}_{j,j}^{-1})    \right )\mathbf{d}^k_m,  \text{ for }\;1 \leq i \leq N, \label{T1c} 
\end{eqnarray}
where $\mathbf{d}^k_m = -(1-\delta_{i,N}) \mathbf{A}_{m,m+1}^{-1} \mathbf{x}_{m+1}^{k-1}$, for $1 \leq m \leq N$.

   \end{theorem}

   \begin{proof}
First, note that the results in equations (\ref{T1a}) and (\ref{T1b}) are the matrix form expressions of results given in equations (\ref{tp0}) and (\ref{tp1}).

Next, for states in $\mathcal{S}(i), 1\leq i \leq N$, and $k\geq 1$, we determine tail probabilities using a first-step argument.  Thus, for a fixed integer $k\geq 1$ and a given state $(i,s,v) \in \mathcal{S}(i)$, with $1\leq i \leq N$, we condition on the first transition out of the initial state and we derive the following relations:
\begin{eqnarray*}
    x^k_{(i,s,v)}&=&\frac{\beta is}{q_{(i,s,v)}}  x^{k-1}_{(i+1,s-1,v)} +\frac{h \beta iv}{q_{(i,s,v)}}  x^{k-1}_{(i+1,s,v-1)} \nonumber\\
    &+&\frac{\gamma i} {q_{(i,s,v)}}  x^k_{(i-1,s,v)}+\frac{\rho s}{q_{(i,s,v)}}  x^k_{(i,s-1,v+1)}\\
     &+&\frac{\theta v}{q_{(i,s,v)}}  x^k_{(i,s+1,v-1)} +\frac{\epsilon (N-i-s-v)}{q_{(i,s,v)}}  x^k_{(i,s+1,v)} , \text{ for } 1\leq i \leq N-1,\nonumber\\
     x^k_{(N,0,0)}&=&\frac{\gamma N}{q_{(N,0,0)}}  x^k_{(N-1,0,0)}.
\end{eqnarray*}
which are equivalent to
\begin{eqnarray}
 -\gamma i x^k_{(i-1,s,v)} - \rho s x^k_{(i,s-1,v+1)} &+& q_{(i,s,v)} x^k_{(i,s,v)}    \label{pri}\\ 
 - \theta v x^k_{(i,s+1,v-1)} & -& \epsilon (N-i-s-v)x^k_{(i,s+1,v)}\nonumber\\
 =\beta is  x^{k-1}_{(i+1,s-1,v)}& +&h \beta iv  x^{k-1}_{(i+1,s,v-1)}
   , \text{ for } 1\leq i \leq N-1,\nonumber \\
    -\gamma N  x^k_{(N-1,0,0)}+q_{(N,0,0)} x^k_{(N,0,0)}&=&0. \label{prN}
\end{eqnarray}
For each level $i$, $1\leq i \leq N$, we can express equations (\ref{pri}) and (\ref{prN}) in matrix form as follows
\begin{eqnarray}
    -\mathbf{A}_{i,i-1} \mathbf{x}_{i-1}^k  -\mathbf{A}_{i,i} \mathbf{x}_{i}^k =(1-\delta_{i,N})  \mathbf{A}_{i,i+1} \mathbf{x}_{i+1}^{k-1} \label{pr3}.
\end{eqnarray}
 Recalling the definition of $d^k_m$ given in the statement of the theorem, equation (\ref{pr3}) is
 \[
  -\mathbf{A}_{i,i-1} \mathbf{x}_{i-1}^k  -\mathbf{A}_{i,i} \mathbf{x}_{i}^k =d^k_i, \text{ for } 1\leq i \leq N \text{ and } k\geq 1.  
 \]
 Finally, a standard forward elimination-backward substitution procedure drives to expression (\ref{T1c}).
   \end{proof}

Since neither the variable number of cases of infection, $L$, nor the auxiliary variables are bounded, we cannot determine their moments directly from their mass distribution functions. Instead, we will derive recursive schemes for computing means, variances, and any other higher-order moment.

First, we introduce some notation. Given $k \geq 0$ we denote by $\mathbf{m}^k_i$ the $c_i-$dimensional vector containing the factorial moments of order $k$ related to states in $\mathcal{S}(i)$. That is, $\mathbf{m}_i^k=(m^k_{(i,s,v)}: (i,s,v) \in \mathcal{S}(i))^{\prime}$.

Following result shows the scheme to compute the factorial moments of auxiliary variables $L_{(i,s,v)}$, for any state $(i,s,v) \in \mathcal{S}$, in a recursive way starting from the explicit values of the order zero moments.

 \begin{theorem} \label{Tmo} 
       Given a non-negative integer $k$, factorial moments of order $k$, $\mathbf{m}_i^k$ are computed as follows
\begin{eqnarray}
\mathbf{m}_i^0 &  = & \mathbf{1}_{c_i}, \text{ for } 0\leq i \leq N  \label{m0}\\
\mathbf{m}_0^k &=& \mathbf{0}_{c_0},\text{ for } k \geq 1 \label{mi0}\\
 \mathbf{m}_N^k&  = & -\mathbf{B}_{N,N}^{-1} \mathbf{C}_N^k, \text{ for } k \geq 1 \label{miN} \\
 \mathbf{m}_i^k&  = &-\mathbf{B}_{i,i}^{-1} (\mathbf{A}_{i,i+1} \mathbf{m}_{i+1}^k + \mathbf{C}_i^k), \text{ for } k \geq 1 \text{ and }  \;1 \leq i \leq N -1, \label{mik} 
 \end{eqnarray}
 where matrices $\mathbf{B}_{i,i}$ and $\mathbf{C}_i^k$ are recursively determined from
 \begin{eqnarray*}
     \mathbf{B}_{i,i}&=&\delta _{1,i} \mathbf{A}_{1,1}+(1-\delta _{1,i})(\mathbf{D }_i^k-\mathbf{A}_{i,i-1}\mathbf{B}_{i-1,i-1}^{-1} \mathbf{A}_{i-1,i}) \label{mB} \\
       \mathbf{C}_{i}^k&=&\delta _{1,i} \mathbf{D}_{1}^k +(1-\delta _{1,i})\left(\mathbf{D }_i^k-\mathbf{A}_{i,i-1}\mathbf{B}_{i-1,i-1}^{-1} \mathbf{C}_{i-1}^k\right) \label{mC}
 \end{eqnarray*}
with $\mathbf{D }_i^k= (1-\delta_{i,N}) k \mathbf{A}_{i,i+1}\mathbf{m}_{i+1}^{k-1}$, for $1 \leq i \leq N.$  
   \end{theorem}
 
\begin{proof}
For states $(0,s,v) \in \mathcal{S}(0)$, the results shown in equations (\ref{m0}) and (\ref{mi0}) come from equation (\ref{m1}).
 
For remaining states, a first-step argument, conditioning on the first transition out of a fixed state $(i,s,v) \in \bigcup_{i=1}^N \mathcal{S}(i) $, gives that generating functions $\varphi_{(i,s,v)}(z)$ satisfy the following system of linear equations:
\begin{align}
    \varphi_{(i,s,v)}(z)&=\frac{\beta is}{q_{(i,s,v)}}z  \varphi_{(i+1,s-1,v)}(z) +\frac{h \beta iv}{q_{(i,s,v)}}z  \varphi_{(i+1,s,v-1)}(z) \nonumber\\
    &+\frac{\gamma i}{q_{(i,s,v)}}  \varphi_{(i-1,s,v)}(z) +\frac{\rho s}{q_{(i,s,v)}}  \varphi_{(i,s-1,v+1)}(z) \nonumber \\
     &+\frac{\theta v}{q_{(i,s,v)}}  \varphi_{(i,s+1,v-1)}(z) +\frac{\epsilon (N-i-s-v)}{q_{(i,s,v)}}  \varphi_{(i,s+1,v)}(z), \nonumber
\end{align}
which are equivalent to
 \begin{eqnarray}
   -\gamma i \varphi_{(i-1,s,v)}(z) - \rho s \varphi_{(i,s-1,v+1)}(z) 
   + q_{(i,s,v)}\varphi_{(i,s,v)}(z) -\theta  v \varphi_{(i,s+1,v-1)}(z)    \nonumber \\
 - \epsilon (N-i-s-v) \varphi_{(i,s+1,v)}(z)
 =\beta is z  \varphi_{(i+1,s-1,v)}(z) +h \beta iv z  \varphi_{(i+1,s,v-1)}(z) \nonumber, \\
 \label{eqphi}
 \end{eqnarray}

We observe that, when $(i,s,v) \in \bigcup_{i=1}^N \mathcal{S}(i)$,
the system of equations given by (\ref{eqphi}) is strictly diagonally dominant. Therefore, for each $|z| \leq 1$ there is a unique  solution of the system. In particular, for $z=1$ we have that the trivial choice $\varphi_{(i,s,v)}(1)=1$, for all $(i,s,v) \in \bigcup_{i=1}^N \mathcal{S}(i)$, is the solution of the above system of equations. Then,
\[
m^0_{(i,s,v)}=\varphi_{(i,s,v)}(1)=1, \text{ for }(i,s,v) \in \bigcup_{i=1}^N \mathcal{S}(i),
\]
which proves result in equation (\ref{m0}), for $1 \leq i \leq N.$

 Next, by differentiating equation (\ref{eqphi}) regarding $z$ repeatedly $k$ times and evaluating at $z=1$ we get a new system of equations involving factorial moments that look as follows:
 \begin{eqnarray}
   -&\gamma& i m^k_{(i-1,s,v)} - \rho s m^k_{(i,s-1,v+1)} + q_{(i,s,v)} m^k_{(i,s,v)}(z) \nonumber \\ & -&  \theta  v m^k_{(i,s+1,v-1)} 
 - \epsilon (N-i-s-v) m^k_{(i,s+1,v)}(z) \nonumber \\
 &=&\beta is \big ( m^k_{(i+1,s-1,v)}+k m^{k-1}_{(i+1,s-1,v)}\big ) 
 +h \beta iv \big ( m^k_{(i+1,s,v-1)}+k m^{k-1}_{(i+1,s,v-1)}\big) \nonumber,
 \end{eqnarray}
which can be expressed in matrix form with the help of the block description of the infinitesimal generator appearing in the expression (\ref{eq:genQ}).
 Hence,  for $k\geq 1$ and $1 \leq i \leq N$, we have
 \begin{eqnarray}
     - \mathbf{A}_{i,i-1} \mathbf{m}^k_{i-1} - \mathbf{A_{i,i}} \mathbf{m}^k_i -
     (1-\delta_{i,N}) \mathbf{A}_{i,i+1} \mathbf{m}^k_{i+1}&= (1-\delta_{i,N}) k \mathbf{A}_{i,i+1} \mathbf{m}^{k-1}_{i+1}. \nonumber\\
      \label{momat}
 \end{eqnarray}
 Note that the right-hand side term of equation (\ref{momat}) agrees with the definition of ${\bf D}^k_i$ in the statement of Theorem \ref{Tmo}. 

 Finally, a standard forward elimination-backward substitution procedure gives the expressions (\ref{miN}) and (\ref{mik}).
 \end{proof}

\section{Numerical Results}

In this section we present some numerical illustrations of the theoretical results. In all the experiments, we consider an isolated population of $N=100$ individuals and we set the recovery rate $\gamma =1.0$, so that the unit time corresponds to the mean time to recovery from the contact disease.

Let us begin by showing results and applications coming from the stationary distribution.
We recall that as time goes on, the population tends towards a pathogen-free situation, where individuals only fluctuate between susceptible and vaccinated compartments.
As we stated in Section \ref{s2.1}, the stationary distribution describes the chance of finding $s$ susceptible and $N-s$ vaccinated individuals, in terms of a Binomial distribution. 
If we focus on the number of vaccinated individuals in the population in the long run, this random variable behaves like a Binomial distribution of $N$ trials and probability of success given by $\frac{\rho}{\rho + \theta}=\frac{\rho/ \theta}{1+{\rho}/{\theta}}$. Therefore, the probability of success remains constant when the booster rate and the waning rate have a constant ratio.
\begin{figure}
    \centering
    \includegraphics[width=0.70\textwidth]{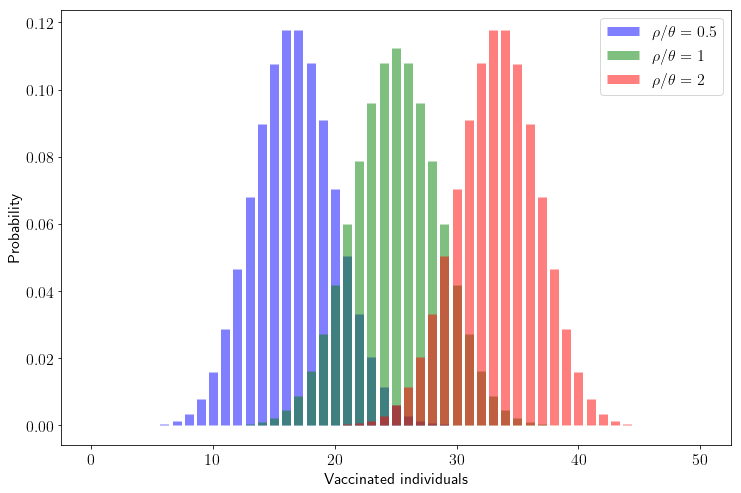}
    \caption{Stationary distribution of $V$, as a function of  $\frac{\rho}{\theta}$}
    \label{fig1nr}
\end{figure}

The Figure \ref{fig1nr} shows the binomial form of the stationary distribution of the number of vaccinated individuals in a population of $N=100$ individuals. We display results for $\frac{\rho}{\theta} \in \{ 0.5, 1.0, 2.0 \}$; that is, when the mean time between booster doses is double, equal to or half the expected duration of vaccine protection. As would be expected, the higher the quotient, the greater the chance that a larger number of vaccinated individuals will be found in the population.

For a general population of $N$ individuals, in the long run, it is expected to find $E[V]=N \frac{\rho / \theta}{1+\rho / \theta}$ vaccinated individuals in average. 
This explicit and simple result may be used to state a criterion for setting an adequate value for the vaccination rate parameter $\rho$, whose inverse value represents the mean time between booster doses of the vaccine. Our criterion will provide protection to susceptible individuals by maintaining a sufficient group of vaccinated individuals, taking into account the characteristics of the disease (basic reproductive number $R_0$) and the characteristics of the vaccine (efficacy and waning effect).
Our idea is to choose the booster rate of the vaccine so that, in the long run, the expected number of vaccinated individuals is greater than the vaccination coverage that provides population herd immunity.

The term herd immunity threshold refers to the critical proportion of immune individuals required to interrupt epidemic transmission in a population.  There is a simple relationship between herd immunity and $R_0$. For a vaccine-preventable contact disease, if a perfect vaccine is available and a fraction $f$ of the population is vaccinated, then the disease will not spread if $(1-f)R_0<1$. In general, the quality of a vaccine is represented by its probability of failure $h \in [0,1]$, where $h=0$ indicates a perfect vaccine and $h=1$ indicates a useless one. 

For imperfect but not useless vaccines, critical coverage is related to the quotient $(1-1/R_0)/(1-h)$ and is the result of a reduction in virus transmission caused by the removal of protected individuals from the susceptible class. 
The critical coverage, $f$, gives the herd-immunity threshold based on the control reproduction number $R_c= R_0(1-(1-h)f)$, which guarantees that $R_c<1$. 
More precisely, starting from the basic reproductive number of the  analogous deterministic model, $R_0=\beta /N \gamma$, the coverage is chosen to satisfy $f> (R_0-1)/R_0(1-h)$).

Using the explicit value of $E[V]$ along with the waning rate of the vaccine, for a given contact disease and an appropriate vaccination coverage $v_c=N f$, our criterion provides the booster rate $\rho$ such that $E[V] \geq v_c$. That is,
\[
\rho \geq \frac{v_c}{N-v_c} \theta.
\]
In the Table \ref{tabcrit}, we show the vaccination coverage and the ratio between the booster rate and the waning rate as we vary the probability of vaccine failure $h \in \{0.05, 0.1, 0.2\}$, for diseases with a basic reproduction ratio $R_0\in \{1.2, 2.5, 4.2\}$. The results in Table 2 for $v_c$ are intuitively correct, showing that we must increase coverage for less effective vaccines (i.e.; as vaccine failure probability increases) and for increasing values of $R_0$. On the other hand, for increasing values of the probability of vaccine failure and also for increasing values of $R_0$, the ratio $\theta / \rho$ decreases. 
\begin{table}
\caption{Vaccination coverage and relation between booster and waning vaccine rates}
       \centering
       \begin{tabular}{|l|l|l|l|}
       \hline
 $v_c$  $\vert \vert$ $\mathbf{\theta / \rho}$ & $h=0.05$& $h=0.1$& $h=0.2$ \\
 \hline
         $R_0=1.2$ & $18$ $\vert \vert$ \textbf{4.55}& $19$ $ \vert \vert$ \textbf{4.26}  & $21$  $\vert \vert$ \textbf{3.76}\\
           $R_0=2.5$ & $64$ $\vert \vert$ \textbf{0.56}  & 67 $\vert \vert$ \textbf{0.49}  & 75 $\vert \vert$ \textbf{0.32}\\
           $R_0=4.2$ & $81$  $\vert \vert$ \textbf{0.23}& $85$ $\vert \vert$ \textbf{0.17}  & $96$  $\vert \vert$ 
 \textbf{0.04}\\ \hline
                \end{tabular}
       \label{tabcrit}
   \end{table}

For example, for the Ebola virus, the transmission $R_0$ is around 1.2. Vaccines have only recently been introduced and there is not yet enough information on the waning effects, but their efficacy varies from 75\% to 100\%. So, in an isolated population of 100 people, we need to keep at least 18 vaccinated if the vaccine is 95\% effective, or at least 21 if the vaccine is 80\% effective. 
Correspondingly, to maintain the herd protection level in the long term, the ratio between booster doses and waning time is 4.55 (i.e.; the mean time between booster doses ($1/\rho$) is at most 4.55 the expected duration of vaccine protection ($1/\theta$)) or 3.76 (if the vaccine is 80\% effective).
For a disease with basic reproduction number $R_0$ close to 4.2, such as diphtheria, if we administer a 95\% effective vaccine, we observe that $v_c=81$ and $\theta / \rho =0.23$.  Therefore, if we fix the number of vaccine protected individuals to be at least 81 in the long run, the mean time between booster doses should be approximately the fourth part of the expected duration of vaccine protection.

\vspace{0.1in}
We will now present results that deal with $L$, the incidence of the infection or the total number of cases of infection that are observed in an outbreak.
We consider a population of 100 individuals, the contact rate $\beta$ and the loss of protection rate $\epsilon$ are fixed at 0.04. Finally, the probability of failure is set at 0.1.

The Table \ref{tab:EVN} shows results for the expected and standard deviation of the incidence $L$. We choose $(\theta ,\rho) \in \{(0.5, 1.0), (1.0,1.0), (1.0, 0.5)   \} $%
, and three different proportions of susceptible and vaccinated individuals as initial situations. More precisely, the initial number of vaccinated individuals is the half, equal to and twice the number of susceptible individuals.

It can be seen that the expected incidence of the infection is reduced when we increase the initial proportion of vaccinated individuals relative to the susceptible ones. For the standard deviation, $SD[L]$, we find that the concentration of the incidence $L$ around its expected value is slightly affected by the initial number of vaccinated individuals.

If we look at the influence of the quotient between the waning rate and the booster rate, ~$\theta / \rho$, numerical results show that the mean and standard deviation of $L$ are not constant when both rates have a constant ratio. For our particular choice of values of the pair $(\theta , \rho)$, we find that an increase in its quotient increases both the expected incidence of cases of infection and the variability around this value.

\begin{table}[htbp]
\centering
\caption{Mean and standard deviation for the total number of cases of infection. }

\begin{tblr} {
  row{2-7} = {c},
  cell{2}{1} = {r=6}{},
  cell{2}{2} = {r=2}{},
  cell{4}{2} = {r=2}{},
  cell{6}{2} = {r=2}{},
  vline{4-7} = {1,3,5,7}{},
  vline{-} = {2-7}{},
  vline{3-7} = {4,6}{},
  hline{1} = {4-6}{},
  hline{2,8} = {-}{},
  hline{4,6} = {2-6}{},
}
                                      &   &   & $(\theta, \rho)=(0.5, 1.0)$  & $(\theta, \rho)=(1.0, 1.0)$  & $(\theta, \rho)=(1.0, 0.5)$  \\
\begin{sideways}$(i_{0}, s_{0}, {v_{0}})$
\end{sideways} & $(1,66,{33})$ & $E[L]$ & 31.7604 & 51.0289 & 62.4891  \\
                                      &   & $SD[L]$ & 33.1116 & 43.7637 & 47.1506 \\
                                      & $(1,49,{50})$ & $E[L]$ & 27.2026 & 46.8978 & 57.9486  \\
                                      &   & $SD[L]$ & 32.0267 & 44.0196 & 48.0731  \\
                                      & $(1,33,{66})$ & $E[L]$ & 23.0980 & 42.8589 & 53.2856 \\
                                      &   & $SD[L]$ & 30.6280 & 43.9308 & 48.5877 
    \end{tblr}

  \label{tab:EVN}%
\end{table}%







%
%
%
\section{Conclusions}

In this paper we consider a stochastic SVIRS model with imperfect vaccine and loss of protection. The mathematical model used to describe how the epidemic evolves is a compartmental stochastic model involving a three-dimensional CTMC.
The finite size and isolation of the population guarantees that the pathogen will disappear as time goes on. And the stationary distribution gives chance to situations where only vaccinated and susceptible individuals exist in the population. The form of this distribution is binomial of N trials and probability of success depending on the quotient between the waning effect rate and the vaccination rate.

Apart from that, our aim in this paper is to analyze the probabilistic nature of $L$, the number of infections observed during an outbreak of the disease.
With the help of some auxiliary variables, we derive stable recursive schemes to compute the mass distribution and moments of the random variable $L$.  

To represent more realistic situations, our model can be extended to non-isolated populations by considering transmission functions $\beta(i,s)$ and $\eta (i,v)$ that are not zero when $i = 0$. In this situation, the disease may disappear for a short time (i.e.; while $I(t)=0$), but through external contacts, the infection is reintroduced later \cite{G-LHAB}. 

The research in this paper can be extended to measure the spread of disease in the population by focusing on the infectiousness of each resident to the whole group  \cite{Art-LHR0,LG17,LH17} or just to the vaccinated pool of individuals \cite{G-LM-LH24}. It will also be interesting to analyse the time taken to reach a threshold in the number of infections. This can be interpreted as a first passage time of the underlying process describing the evolution of epidemics (see e.g. \cite{G-LHMa,GC-LG-LH-T}).

\section{Acknowledgments}
This research was supported by the Ministry of Science and Innovation (Government of Spain) through project PID2021-125871NB-I00. The second author also acknowledges the support of Banco Santander and the Complutense University of Madrid, Pre-doctoral Researcher Contract CT63/19-CT64/19.

\section{Appendix: Supplemental material}

In this Appendix, we characterize the infinitesimal generator ${\bf Q}$ of the three-dimensional CTMC ${\cal X}$ with state space ${\cal S}$ 
\begin{equation*}
    \mathcal{S}=\{(i,s,v): 0\leq i,s,v\leq 0, \quad i+s+v\leq N \}.
\end{equation*}

Note that $\mathcal{S}$ containing $\binom{N+3}{3}$ states, that can be organized into levels according to the number of infected individuals, as follows:
\begin{equation*}
 {\cal S}= \cup_{i=0}^{N} S(i)   
\end{equation*}
where, for $0\leq i\leq N$, each set $S(i)$ is  partitioned in terms of sub-levels as $ S(i)= \cup_{s=0}^{N-i} l(i,s)$ with $ l(i,s)=\{(i,s,v): 0 \leq v \leq N-i-s \}$. In particular, the order of sub-levels inside $S(i)$ is
\begin{equation*}
     l(i,0) \prec l(i,1) \prec \ldots \prec l(i, N-i),
\end{equation*}
while we order the states in $l(i,s)$ according to the number of vaccine protected individuals. That is 
\begin{equation*}
    (i,s,0) \prec (i,s,1) \prec \dots \prec (i,s,N-i-s).
\end{equation*}

\noindent Thus, the number of states in each level $S(i)$ is given by $\# S(i) = \binom{N - i + 2}{2}$ and each sub-level $l(i,s)$ contains $\#l(i,s)=N-i-s+1$, for $0\leq i\leq N$ and $0\leq s\leq N-i$.

The above partition of the state space ${\cal S}$ yields a block-tridiagonal structure for the infinitesimal generator of ${\cal X}$ that looks as follows: 
\begin{eqnarray*}
    {\bf Q} &=& \left(\begin{array}{ccccc}
        {\bf A}_{0,0} & {\bf A}_{0,1} & & &
        \\
        {\bf A}_{1,0} & {\bf A}_{1,1} & {\bf A}_{1,2} & &
        \\
                      & \ddots        & \ddots        & \ddots &
        \\
                      &               & {\bf A}_{N-1,N-2} & {\bf A}_{N-1,N-1} & {\bf A}_{N-1,N}
        \\
                      &               &                   & {\bf A}_{N,N-1}   & {\bf A}_{N,N}
        \end{array}\right),
        \label{eq:genQ1}
\end{eqnarray*}
where the entries of the sub-matrices $A_{i,i*}$ are the transition rates from states in $S(i)$ to states in $S(i*)$. Specifically, the non-null elements are described as follows:

\begin{itemize}[label=\textbullet]
    \item For $1\leq i\leq N$, sub-matrix $\mathbf{A}_{i,i-1}$ is associated with jumps of process ${\cal X}$ from states in level $S(i)$ to states in level $S(i-1)$. In more detail, it has the following form
    \begin{eqnarray*}
    {\bf A}_{i,i-1} &=& \left(\begin{array}{ccccc}
    {\bf B}_{i,i-1}(0,0) & & & & \\
     & {\bf B}_{i,i-1}(1,1) & & &  \\
     & & \ddots & & \\
     & & & {\bf B}_{i,i-1}(N-i,N-i) & 0
     \end{array}\right).
    \end{eqnarray*}
    Matrices $\mathbf{B}_{i,i-1}(s,s)$ record transition rates from sub-level $l(i,s)$ to sub-level $l(i-1,s)$. We express the particular case of $i=N$ as
    \begin{eqnarray*}
    {\bf A}_{N,N-1} &=& \left(\begin{array}{cc}
    {\bf B}_{N,N-1}(0,0), & 0
    \end{array}\right).
    \end{eqnarray*}

    \item For $0\leq i\leq N$, sub-matrix $\mathbf{A}_{i,i}$ is related to transitions from level $S(i)$ to level $S(i)$, and is given by
    \begin{center}
    \resizebox{1\textwidth}{!}{$
    \mathbf{A}_{i,i} = \left(\begin{array}{ccccc}
    {\bf B}_{i,i}(0,0) & {\bf B}_{i,i}(0,1) & & & \\
    {\bf B}_{i,i}(1,0) & {\bf B}_{i,i}(1,1) & {\bf B}_{i,i}(1,2) & & \\
     & \ddots & \ddots & \ddots & \\
     & & {\bf B}_{i,i}(N-i-1,N-i-2) & {\bf B}_{i,i}(N-i-1,N-i-1) & {\bf B}_{i,i}(N-i-1,N-i) \\
     & & & {\bf B}_{i,i}(N-i,N-i-1) & {\bf B}_{i,i}(N-i,N-i)
     \end{array}\right).
    $}
    \end{center}
    Matrices ${\bf B}_{i,i}(s,s')$ record transition rates from sub-level $l(i,s)$ to sub-level $l(i,s')$, for $s-1\leq s'\leq s+1$. It should be pointed out the particular case $i=N$ which leads to 
    \begin{equation*}
    {\bf A}_{N,N}= \left(-q(N,0,0) \right), 
    \end{equation*}
    where $q(N,0,0)=\gamma N$.\\

    \item For $0\leq i \leq N-1$, sub-matrix $\mathbf{A}_{i,i+1}$, corresponds to transitions from states in level $S(i)$ to states in level $S(i+1)$, and is given by
    \begin{center}
    \resizebox{1.1\textwidth}{!}{$
    \mathbf{A}_{i,i+1} = \left(\begin{array}{ccccc}
    {\bf B}_{i,i+1}(0,0) &  & & & \\
    {\bf B}_{i,i+1}(1,0) & {\bf B}_{i,i+1}(1,1) &  & & \\
    & {\bf B}_{i,i+1}(2,1) & {\bf B}_{i,i+1}(2,2)  & & \\
     &  & \ddots & \ddots & \\
     & &  & {\bf B}_{i,i+1}(N-i-1,N-i-2) & {\bf B}_{i,i+1}(N-i-1,N-i-1) \\
     & & &  & {\bf B}_{i,i+1}(N-i,N-i-1)
     \end{array}\right),
    $}
    \end{center}
    where matrices ${\bf B}_{i,i+1}(s,s')$ are related to transition rates from sub-level $l(i,s')$ to sub-level $l(i+1,s')$, for $s'=s-1,s$.
\end{itemize}

\noindent The above matrices ${\bf B}_{i,i'}(s,s')$, for $i-1\leq i'\leq i+1$, and $s-1\leq s'\leq s+1$,
are described below.

\begin{itemize}[label=\textbullet]
    \item For $1\leq i\leq N$ and $0\leq s\leq N-i$, matrix $ \mathbf{B}_{i,i-1}(s,s)$ contains transitions rates from states in level $l(i,s)$ to states in level $l(i-1,s)$, and is given by
    \begin{equation*}
        \mathbf{B}_{i,i-1}(s,s)=\left(\gamma i{\bf I}_{N-i-s+1} \quad {\bf 0}_{N-i-s+1}\right), 
        \end{equation*}
    where ${\bf I}_{a}$ is the identity matrix of order $a$.
    \vspace{0.2cm}
    \item For $0\leq i\leq N$ and $1\leq s\leq N-i$, matrix ${\bf B}_{i,i}(s,s-1)$ contains transitions rates from states in level $l(i,s)$ to states in level $l(i,s-1)$, and is given by
    \begin{equation*}
    \mathbf{B}_{i,i}(s,s-1)=\left( {\bf 0}_{N-i-s+1} \quad \rho s{\bf I}_{N-i-s+1}\right).
    \end{equation*}

    \item For $0\leq i\leq N$ and $0\leq s\leq N-i$, matrix $\mathbf{B}_{i,i}(s,s)$ contains transitions rates from states in level $l(i,s)$ to states in level $l(i,s)$, and is given by
    \begin{equation*}
        \mathbf{B}_{i,i}(s,s)=diag\left(-q_{(i,s,v)}: 0\leq v\leq N-i-s \;\right),
    \end{equation*}
   with $q(i,s,v)= \beta(i,s) + \eta(i,v) + \gamma(i) + \rho s + \theta v + \epsilon(N-i-s-v)$, where $diag(a_{1}, a_{2},..., a_{m})$ is a diagonal matrix with non-null entries $a_{1}, a_{2},..., a_{m}$.
   
\vspace{0.3cm}
    \item For $0\leq i\leq N$ and $0\leq s\leq N-i-1$, matrix $\mathbf{B}_{i,i}(s,s+1)$ contains transitions rates from states in level $l(i,s)$ to states in level $l(i,s+1)$, and is given by
    \begin{center}
    \resizebox{1.1\textwidth}{!}{$
    \mathbf{B}_{i,i}(s,s+1) = \left(
    \begin{array}{cccccc}
    \epsilon(N-i-s) &  &  &  &  &  \\
    \theta & \epsilon(N-i-s-1) &  &  &  &  \\
    & 2\theta & \epsilon(N-i-s-2) &  &  &  \\
     &  & \ddots & \ddots &  &  \\
    &  &  & (N-i-s-2)\theta & 2\epsilon &  \\
    &  &  &  & (N-i-s-1)\theta & \epsilon \\
    &  &  &  &  & (N-i-s)\theta \\
    \end{array}
    \right).
    $}  
    \end{center}
\vspace{0.2cm}
    \item For $0\leq i \leq N-1$ and $1\leq s \leq N-i-1$, matrix $\mathbf{B}_{i,i+1}(s,s-1)$ contains transitions rates from states in level $l(i,s)$ to states in level $l(i+1,s-1)$, and is given by
    \begin{equation*}
        \mathbf{B}_{i,i+1}(s,s-1)=\left( \beta is \; {\bf I}_{N-i-s+1}  \right).
    \end{equation*}    
\vspace{0.2cm}
    \item For $0\leq i \leq N-1$ and $0\leq s \leq N-i-1$, matrix $\mathbf{B}_{i,i+1}(s,s)$ contains transitions rates from states in level $l(i,s)$ to states in level $l(i+1,s)$, and is given by
    \begin{eqnarray*}
    \mathbf{B}_{i,i+1}(s,s) &=& \left(\begin{array}{c}
    {\bf 0}'_{N-i-s}  \\
     diag\left(\beta hiv: 1\leq v\leq N-i-s\;\right)
     \end{array}\right).
     \end{eqnarray*}
    
    \end{itemize}

\end{document}